\documentclass[11 pt]{article}
\renewcommand{\baselinestretch}{1.4}
\date{}

\usepackage[section]{placeins}
\usepackage{graphicx,subfigure}
\usepackage{amsmath,amssymb}
\usepackage{cite}
\usepackage{url}
\usepackage{color}
\usepackage{dblfloatfix}
\usepackage{mystyle}

\begin{document}

\title{Multiuser Diversity Gain in Cognitive Networks}


\author{Ali~Tajer~\and~Xiaodong~Wang
\thanks{The authors are with the Department of Electrical Engineering, Columbia University, New York, NY 10027 (email:\{tajer,  wangx\}@ee.columbia.edu).}}

\maketitle

\allowdisplaybreaks
\begin{abstract}
Dynamic allocation of resources to the \emph{best} link in large
multiuser networks offers considerable improvement in spectral
efficiency. This gain, often referred to as \emph{multiuser diversity gain}, can be cast as double-logarithmic growth of the network
throughput with the number of users. In this paper we consider large
cognitive networks granted concurrent spectrum access with  license-holding users. The primary network affords to share its under-utilized spectrum bands with the secondary users. We assess the optimal multiuser diversity gain in the cognitive networks by quantifying how the sum-rate throughput of the network scales with the number of secondary users. For this purpose we look at the optimal pairing of spectrum bands and secondary users, which is supervised by a central entity fully aware of the instantaneous channel conditions, and show that the throughput of the cognitive network scales double-logarithmically with the number of secondary users ($N$) and linearly with the number of available spectrum bands ($M$), i.e., $M\log\log N$. We then propose a \emph{distributed} spectrum allocation scheme, which does not necessitate a central controller or any information exchange between different secondary users and still obeys the optimal throughput scaling law. This scheme requires that \emph{some} secondary transmitter-receiver pairs exchange $\log M$ information bits among themselves. We also show that the aggregate amount of information exchange between secondary transmitter-receiver pairs is {\em asymptotically} equal to $M\log M$. Finally, we show that our distributed scheme guarantees fairness among the secondary users, meaning that they are equally likely to get access to an available spectrum band.
\end{abstract}

{\bf kewords:}
Cognitive radio, distributed, fairness, multiuser diversity, spectrum allocation.

\section{Introduction}
\label{sec:intro}

Dense multiuser networks offer significant spectral efficiency
improvement by dynamically identifying and allocating the communication resources to  the \emph{best} link. The improvements thus attained are often referred to as \emph{multiuser diversity gain} and rest on the basis of opportunistically allocating all the resources to the most reliable link. The performance of such resource allocation scheme relies on the peak, rather than average, channel conditions and improves as the number of users increases, as it becomes more likely to have a user with an instantaneously strong link.

The notion of opportunistic communication and multiuser diversity was first introduced \cite{Knopp:ICC95} for uplink transmissions, and further developed in~\cite{Tse:ISIT97,Viswanath:IT02,Sanayei:IT07} for
downlink transmissions. The analysis of multiuser diversity gain in
downlink multiple-input multiple-output (MIMO) channels is provided in
\cite{Sharif:IT05, Sharif:COM07}. In all these transmission schemes,
the sum-rate capacity exhibits a double-logarithmic growth with the
number of users.

The recent advances in secondary spectrum leasing \cite{FCC} and
cognitive networks \cite{Mitola:thesis} suggest accommodating unlicensed users (secondary users or cognitive radios), within the license-holding networks and allowing them to access under-utilized spectrum bands. Among different spectrum sharing schemes, {\em underlaid} spectrum  access~\cite{Goldsmith:IEEE09} has received significant attention. This scheme allows for simultaneous spectrum access by the primary and secondary users, provided that the power of secondary users is controlled such that they impose limited interference to the primary users.

In this paper we consider opportunistic underlaid spectrum access by secondary users and assess the multiuser diversity gain by
analyzing the sum-rate throughput scaling of the cognitive network. Such
analysis for cognitive networks differs from those of primary networks
studied in \cite{Knopp:ICC95, Tse:ISIT97, Viswanath:IT02, Sharif:IT05,
Sharif:COM07} in two directions. First, the transmissions in the
cognitive network are contaminated by the interferences induced by the
primary users. The existence of such interference does not make
opportunistic communication possible by merely finding the strongest
secondary link, and necessitates accounting for the effect of
interference as well. Secondly, and more importantly, the  uplink/downlink transmissions in the networks referenced above, require feedback from the users to the base station and it is the base station that dynamically decides which user(s) should receive the resources. Cognitive networks, in contrary, are often assumed to lack any infrastructure or central entity and spectrum allocation should be carried out in a distributed way.

To address these two issues, we first focus on only examining the for the effects of interference and assume that the cognitive network has a central decision-making entity, fully aware of all cognitive users' instantaneous channel realizations. This result, providing the optimal scaling factor, presents an upper bound on the throughput yielded by any distributed spectrum allocation scheme. In the next step, we offer a \emph{distributed} algorithm where the secondary users decide about accessing a channel merely based on their own perception of  instantaneous network conditions.

Our analyses reveal that, interestingly, in both centralized and
distributed setups, the sum-rate throughput scales double-logarithmically with the number of users, which is the optimal growth and is the same as that of centralized primary networks. Therefore, the interference from the primary network incurs no loss on the multiuser diversity gain of the cognitive network.

We also examine how fairness is maintained in our distributed scheme.
Generally, in opportunistic communication schemes there exists a
conflict between fairness and multiuser diversity gain, as the network
tends to reserve the resources for the most reliable links, which leads the network to be dominated by the users with strong links. We show
that, however, in our distributed scheme, we can ensure fairness among
the secondary users by providing them with the same opportunity for accessing an available spectrum band.

The remaining part of the paper is organized as follows. In Section
\ref{sec:descriptions} we provide the system model as well as the statement of the problem.   Sections~\ref{sec:centralized}~and~\ref{sec:distributed} discuss the
sum-rate throughput scaling laws in centralized and distributed
cognitive networks, respectively. Our distributed algorithm requires
some information exchange between each cognitive transmitter and its
designated receiver. The amount of such information is quantified in
Section~\ref{sec:information}. As we are considering opportunistic type of spectrum access, it is crucial to also look at the fairness among
the users. The discussion on the fairness is given in
Section~\ref{sec:fairnes}. Some remarks on the implementation of the distributed spectrum access algorithm are provided in Section \ref{sec:discussions} and Section~\ref{sec:conclusion} concludes
the paper. In order to enhance the flow of the material, we have
confined most of the proofs in the appendices.

\section{System Descriptions}
\label{sec:descriptions}
\subsection{System Model}
\label{sec:model}

We consider a \emph{decentralized} cognitive network comprising of $N$
secondary transmitter-receiver pairs coexisting with the primary
transmitters via \emph{underlaid}~\cite{Goldsmith:IEEE09}
spectrum access. Therefore, the primary and secondary users can coexist simultaneously on the same spectrum band. The primary network
affords to accommodate $1\leq M\ll N$ secondary users and allows them
to access the non-overlapping spectrum bands $B_1,\dots, B_{M}$ such
that each band is allocated to \emph{exactly one} secondary
transmitter-receiver pair. We assume that the secondary transmitters and receivers are paired up {\em a priori} such that each secondary  transmitter knows its designated receiver and vice versa. We also assume that each secondary transmitter and receiver is potentially capable of operating on each of the $M$ spectrum  bands, a feature facilitated by having appropriate reconfigurable hardware.

We assume quasi-static flat fading channels and denote the channel from the $j^{th}$ primary transmitter to the $i^{th}$ secondary receiver in the $m^{th}$ spectrum band ($B_m$) by $h^m_{i,j}\in\mathbb{C}$ and denote the channel between the $i^{th}$ secondary transmitter-receiver pair in the $m^{th}$ spectrum band ($B_m$) by $g^m_i\in\mathbb{C}$. Let
$x^p_i(t)$ and $x^s_i(t)$ represent the transmitted signals by the
$i^{th}$ primary transmitter and the $i^{th}$ secondary transmitter,
respectively. We assume that there might be a group of active primary
users on each spectrum band $B_m$ and define the set $\mathcal{B}_m$
such that it contains the indices of such users. If the $n^{th}$ secondary pair transmits on $B_m$, then the received signal at the $n^{th}$ secondary receiver is given by
\begin{equation}\label{eq:model1}
    y^m_n=\sqrt{\eta_n}g^m_nx^s_n+\sum_{j\in\mathcal{B}_m}\sqrt{\gamma_{n,j}}h^m_{n,j}x^p_j+z^m_n,
\end{equation}
where $z^m_n\in\mathcal{CN}(0,N_0)$ is the additive white Gaussian
noise at the $n^{th}$ receiver. In a non-homogeneous network, the users experience different path-loss and shadowing effects, which we account
for by incorporating the terms $\{\gamma_{i,j}\}$ and $\{\eta_i\}$.
Also, we assume that the primary and secondary transmitters satisfy
average power constraints $P_p$ and $P_s$, respectively, i.e.,
$\bbe[|x^p_i|^2]\leq P_p$ and $\bbe[|x^s_i|^2]\leq P_s$ and the channel
coefficients $\{h^m_{i,j}\}_{i,j,m}$ and $\{g^m_i\}_{i,m}$ are
i.i.d. and distributed as complex Gaussian $\mathcal{CN}(0,1)$. Each
secondary receiver treats all undesired signals (interference from the
primary users) as Gaussian interferers. Therefore, the
signal-to-interference-plus-noise-ratio ($\sinr$) of the $n^{th}$
secondary pair on the spectrum band $B_m$ is given by
\begin{eqnarray}\label{eq:sinr}
    \nonumber \sinr_{m,n}&=&
    \frac{\eta_n\bbe[|g^m_nx^s_n|^2]}{N_0+\sum_{j\in\mathcal{B}_m}\gamma_{n,j}\bbe[|x^p_jh^m_{n,j}|^2]}\\
    &=&
    \frac{P_s\eta_n|g^m_n|^2}{N_0+P_p\sum_{j\in\mathcal{B}_m}\gamma_{n,j}|h^m_{n,j}|^2}.
\end{eqnarray}
We define the transmission signal-to-noise ratio ($\snr$) by
$\rho\dff\frac{P_s}{N_0}$. Throughout the paper we say that $a_N$ and
$b_N$ are \emph{asymptotically} equal, denoted by $a_N\doteq b_N$ if
$\lim_{N\rightarrow\infty}\frac{a_N}{b_N}=1$, and define $\dotlt$ and
$\dotgt$, accordingly. We also define the set of secondary users
indices by $\mathcal{N}=\{1,\dots,N\}$. All the rates in the paper are in bits/sec/Hz and $\log$ refers to the logarithm in base 2.

\subsection{Problem Statement}
\label{sec:definition}

Our goal is to assess the multiuser diversity gain of cognitive networks. For this purpose we identify $M$ secondary transmitter-receiver pairs out of $N$ available ones and assign one spectrum band $B_m$ to each of them, such that the cognitive network throughput is maximized. We assume that all spectrum bands $B_m$ are of the same bandwidth. Therefore, the maximum throughput is given by
\begin{equation}\label{eq:Rmax}
    R_{\max}=\bbe\left[\max_{A\subset\mathcal{N},\;|A|=M}
    \sum_{m=1}^{M}\log\left(1+\sinr_{m,{A_m}}\right)\right],
\end{equation}
where $A_m$ denotes the $m^{th}$ element of set $A$, for $m=1,\dots,M$
and the maximization is taken over all \emph{ordered} subsets of
$\mathcal{N}$.

In order to find the optimal multiuser diversity gain in the cognitive network we first consider a centralized setup. We assume that there exists a
central decision-making entity in the cognitive network, which is fully and instantaneously aware of the channel conditions of all secondary
users. The central node solves the problem cast in (\ref{eq:Rmax}) by
an exhaustive search for pairing up $M$ secondary users
with the $M$ available channels. For such secondary user-channel pairs
we analyze how the sum-rate of the cognitive network scales as the
number of cognitive users ($N$) increases. Such centralized setup
imposes extensive information exchange\footnote{$M$ real numbers per
user.} which can be prohibitive in large network sizes.

Next, motivated by alleviating the amount of information exchange imposed by the centralized setup and noting that our cognitive network is \emph{ad-hoc} in nature and lacks a central-decision making entity, we propose a decentralized spectrum allocation scheme. In the distributed scheme each secondary user decides about taking over a channel solely based on its own perception of the network realization. We prove that the proposed distributed scheme retains the same throughput scaling law as in the centralized setup, i.e., is asymptotically optimal.

\section{Centralized Spectrum Allocation}
\label{sec:centralized}

The central decision-making unit has access to all $\{\sinr_{m,n}\}$
and performs an exhaustive search over all possible
user-channel (spectrum band) combinations in order to find the one that maximizes the sum-rate throughput given in (\ref{eq:Rmax}). In order to find the throughput scaling, we establish lower and upper bounds on $R_{\max}$, denoted by $R^l_{\max}$ and $R^u_{\max}$, respectively, and show that these bounds are asymptotically equal, i.e.,  $R^l_{\max}\doteq R^u_{\max}$, which in turn provide the optimal throughput scaling law of the cognitive network.

We define the most favorable user of the $m^{th}$ spectrum band as the
user with the largest $\sinr$ on this band, i.e.,
\begin{equation}\label{eq:favorable}
    n^*_m\dff\arg\max_{1\leq n\leq N}\sinr_{m,n}.
\end{equation}
In general, it might so happen that one user is the most favorable user for two different spectrum bands, i.e., $n^*_m=n^*_{m'}$, while $m\neq
m'$, and as a result, these two spectrum bands cannot be allocated to
their most favorable users simultaneously (we have assumed that each
user may get access to only one spectrum band). Let us define
$\mathcal{D}$ as the event that different spectrum bands have distinct
most favorable users i.e., no single user is the most favorable user
for two distinct spectrum bands. Note that pairing the secondary users and the spectrum bands conditioned on the event $\mathcal{D}$ is equivalent to allocating each channel to its most favorable user, i.e.,
\begin{align}
    \nonumber &\bbe\left[\max_{A\subset\mathcal{N},\;|A|=M}\sum_{m=1}^{M}\log\left(1+\sinr_{m,{A_m}}\right)\bigg
    |\;\mathcal{D}\right]\\
    \label{eq:D1}&\hspace{1.2 in}=\bbe\left[\sum_{m=1}^{M}\log
    \left(1+\sinr_{m,n^*_m}\right)\right].
\end{align}
On the other hand, under event $\mathcal{\bar D}$, at least one
spectrum band will not be allocated its most favorable user and
therefore we have
\begin{align}
    \nonumber &\bbe\left[\max_{A\subset\mathcal{N},\;|A|=M}\sum_{m=1}^{M}\log\left(1+\sinr_{m,{A_m}}\right)\bigg
    |\;\mathcal{\bar D}\right]\\
    \label{eq:D2}&\hspace{1.2 in}\leq \bbe\left[\sum_{m=1}^{M}\log
    \left(1+\sinr_{m,n^*_m}\right)\right].
\end{align}
Equations (\ref{eq:D1}) and (\ref{eq:D2}) give rise to
\begin{align}\label{eq:RUmax}
    \nonumber
    R_{\max}&=\bbe\left[\max_{A\subset\mathcal{N},\;|A|=M}\sum_{m=1}^{M}\log
    \left(1+\sinr_{m,{A_m}}\right) \;\bigg
    |\;\mathcal{D}\right]P(\mathcal{D}) \\
    \nonumber &+\bbe\left[\max_{A\subset\mathcal{N},\;|A|=M}\sum_{m=1}^{M}\log
    \left(1+\sinr_{m,{A_m}}\right)\;\bigg
    |\;\mathcal{\bar D}\right]P(\mathcal{\bar D})\\
    &\leq \bbe\left[\sum_{m=1}^{M}\log
    \left(1+\sinr_{m,n^*_m}\right)\right]\dff R^u_{\max}.
\end{align}
Also it can be readily shown that
\begin{align}
    \nonumber R_{\max}&\geq\bbe\left[\max_{A\subset\mathcal{N},\;|A|=M}\sum_{m=1}^{M}\log
    \left(1+\sinr_{m,{A_m}}\right)\;\bigg|\;\mathcal{D}\right]P(\mathcal{D})\\
    \label{eq:RLmax}&=R^u_{\max}P(\mathcal{D})\dff R^l_{\max}.
\end{align}
\begin{lemma}
\label{lemma:D}

$R^l_{\max}$ and $R^u_{\max}$ are asymptotically equal, i.e.,
$R^l_{\max}\doteq R^u_{\max}$.
\end{lemma}
\begin{proof}
See Appendix \ref{app:lemma:D}.
\end{proof}

Now, we find how $R_{\max}^u$ scales as $N$ increases. Note that the
$\sinr$s are statistically independent for all users and spectrum
bands. The reason is that $\sinr_{m,n}$ given in (\ref{eq:sinr})  inherits its randomness from the randomness of $g^m_n$ (fading coefficient of the channel between the $n^{th}$ secondary pair on the $m^{th}$ band) and $\{h^m_{n,j}\}_j$ (the fading coefficient of the channels from the $j^{th}$ primary user to the $n^{th}$ secondary receiver on the $m^{th}$ band). Since for any two different pairs of $(m,n)\neq(m',n')$, the fading coefficients $g^m_n$ and $g^{m'}_{n'}$ refer to fading in different locations or in different spectrum bands, therefore they are statistically independent. Similarly it can be argued that $h^m_{n,j}$ and $h^{m'}_{n',j}$ are also statistically independent for $(m,n)\neq(m',n')$. As a result, all the random ingredients of $\sinr_{m,n}$ and $\sinr_{m',n'}$ for $(m,n)\neq(m',n')$ are independent which in turn justifies the independence of the $\sinr$s. Nevertheless, $\sinr$s are not identically distributed since different users experience different path-losses and shadowing effects. Hence, for more mathematical tractability we build two other sets whose elements provide lower and upper bounds on $\sinr_{m,n}$ and are i.i.d. For this purpose we define
\begin{align*}
  &\gamma_{\max}\dff\max_{i,j}\left\{\frac{\gamma_{i,j}}{\eta_i}\right\}, \;\;\; &&\eta_{\max}=\max_i\eta_i,\\
  \mbox{and}\quad\quad & \gamma_{\min}\dff\min_{i,j}\left\{\frac{\gamma_{i,j}}{\eta_i}\right\},\;\;\;&&\eta_{\min}=\min_i\eta_i.
\end{align*}
For $m=1,\dots, M$ we also define the sets
$\mathcal{S}_l(m)=\{S_l(m,n)\}_{n=1}^N$ and
$\mathcal{S}_u(m)=\{S_u(m,n)\}_{n=1}^N$ such that for $n=1, \dots, N$
\begin{align}
\label{eq:Sl}S_l(m,n)&\dff\frac{|g^m_n|^2}{\frac{1}{\rho\eta_{\min}}+
\frac{P_p}{P_s}\gamma_{\max}\sum_{j\in\mathcal{B}_m}|h^m_{n,j}|^2},\\
\label{eq:Su}\mbox{and}\;\;\;S_u(m,n)&\dff\frac{|g^m_n|^2}{\frac{1}{\rho\eta_{\max}}+
\frac{P_p}{P_s}\gamma_{\min}\sum_{j\in\mathcal{B}_m}|h^m_{n,j}|^2}.
\end{align}
It can be readily verified that $S_l(m,n)\leq\sinr_{m,n}\leq
S_u(m,n)$. We use the notations $\mathcal{S}^{(i)}_l(m)$ and
$\mathcal{S}^{(i)}_u(m)$ to refer to the $i^{th}$ largest elements of the sets $\mathcal{S}_l(m)$ and  $\mathcal{S}_u(m)$, respectively, and use $\sinr^{(i)}_m$ to denote the $i^{th}$ largest element of
$\{\sinr_{m,n}\}_{n=1}^N$. In the following lemma we show how these
ordered elements are related.

\begin{lemma}
\label{lemma:order}

For any spectrum band $B_m$ and any $i=1,\dots, N$
we have $\mathcal{S}^{(i)}_l(m)\leq \sinr^{(i)}_m\leq
\mathcal{S}^{(i)}_u(m).$
\end{lemma}
\begin{proof}
See Appendix \ref{app:lemma:order}.
\end{proof}

Now, by recalling the definition of $R_{\max}^u$ given in
(\ref{eq:RUmax}) and noting that $\sinr_{m,n^*_m}=\sinr^{(1)}_m$ and by
invoking the result of Lemma~\ref{lemma:order} we get
\begin{align}
    \label{eq:RUmax_bounds1} &R_{\max}^u\geq \bbe\left[\sum_{m=1}^{M}\log
    \left(1+\mathcal{S}^{(1)}_l(m)\right)\right], \\
    \label{eq:RUmax_bounds2}\mbox{and}\;\;\;&R_{\max}^u\leq \bbe\left[\sum_{m=1}^{M}\log \left(1+\mathcal{S}^{(1)}_u(m)\right)\right].
\end{align}
In order to further simplify the bounds on $R_{\max}^u$ given in
(\ref{eq:RUmax_bounds1})-(\ref{eq:RUmax_bounds2}), in the following lemma we provide the cumulative density functions
(CDF) of $S_l(m,n)$ and $S_u(m,n)$ (\ref{eq:Sl}) and (\ref{eq:Su}).

\begin{lemma}
\label{lemma:CDF}

The elements of $\mathcal{S}_l(m)$ and $\mathcal{S}_u(m)$ are i.i.d.
and their CDFs are
\begin{align}
\label{eq:CDFl}S_l(m,n) &\sim F_l(x;m)\dff
1-\frac{e^{-x/\rho\eta_{\min}}}{\left(\frac{P_p}{P_s}\gamma_{\max}x+1\right)^{K_m}},\\
\label{eq:CDFu}\mbox{and}\;\;  S_u(m,n) &\sim F_u(x;m)\dff
1-\frac{e^{-x/\rho\eta_{\max}}}{\left(\frac{P_p}{P_s}\gamma_{\min}x+1\right)^{K_m}},
\end{align}
where $K_m\dff |\mathcal{B}_m|$.
\end{lemma}
\begin{proof}
See Appendix \ref{app:lemma:CDF}.
\end{proof}

We denote the  $i^{th}$ order statistics of the statistical samples
$\mathcal{S}_l(m)$ and $\mathcal{S}_u(m)$ with parent distributions
given in (\ref{eq:CDFl})-(\ref{eq:CDFu}) by $\mathcal{S}^{(i)}_l(m)$
and $\mathcal{S}^{(i)}_u(m)$, respectively. By denoting the CDF of
$\mathcal{S}^{(i)}_l(m)$ by $F^{(i)}_l(x;m)$ and that of
$\mathcal{S}^{(i)}_u(m)$ by $F^{(i)}_u(x;m)$, for $i=1,\dots, N$ we
have ~\cite{Arnold:Book}
\begin{align}
  \label{eq:CDF:Lj} F_l^{(i)}(x;m) &= \sum_{j=0}^{i-1}{N\choose
  j}\Big(F_l(x;m)\Big)^{N-j}\Big(1-F_l(x;m)\Big)^j,\\
  \label{eq:CDF:Uj}F_u^{(i)}(x;m) &= \sum_{j=0}^{i-1}{N\choose
  j}\Big(F_u(x;m)\Big)^{N-j}\Big(1-F_u(x;m)\Big)^j.
\end{align}
By invoking the above definitions, (\ref{eq:RUmax_bounds1}) and (\ref{eq:RUmax_bounds2}) can be
re-written as
\begin{align}
    \label{eq:RUmax_bounds2_1}&R_{\max}^u\geq\sum_{m=1}^{M}\int_0^\infty\log(1+x)\;dF_l^{(1)}(x;m),\\
    \label{eq:RUmax_bounds2_2}\mbox{and}\;\;&R_{\max}^u\leq\sum_{m=1}^{M}\int_0^\infty\log(1+x)\;dF_u^{(1)}(x;m).
\end{align}
We also define
\begin{equation}\label{eq:G}
    G(x)\dff1-e^{-x},
\end{equation}
and let $G^{(i)}(x)$ denote the CDF of the $i^{th}$ order statistic of
a statistical sample with $N$ members and with parent distribution
$G(x)$. By using this definition we offer the following lemma which is
a key step in finding how $R^u_{\max}$ scales with increasing $N$.
\begin{lemma}
\label{lemma:G} For the distributions $F_l^{(1)}(x;m)$, $F_u^{(1)}(x;m)$
and $G^{(1)}(x)$ we have
\begin{align}
  \nonumber\int_0^\infty\log(1+x)\;&dF_u^{(1)}(x;m)  \\
  \label{eq:lemma:G1}&\leq\int_0^\infty\log(1+\rho\eta_{\max}
  x)\;dG^{(1)}(x),\\
  \nonumber\mbox{and}\;\;\int_0^\infty\log(1+x)\;&dF_l^{(1)}(x;m)  \\
  \nonumber&\geq\int_0^\infty\log(1+\rho\eta_{\min}x)\;dG^{(1)}(x)\\
  \label{eq:lemma:G2}&-\log\bigg[1+\frac{K_mP_p}{P_s}\gamma_{\max}\rho\eta_{\min}\bigg]. \end{align}
\end{lemma}
\begin{proof}
By using the definitions of $F_l^{(1)}(x;m)$ and $F_u^{(1)}(x;m)$ given
in (\ref{eq:CDF:Lj})-(\ref{eq:CDF:Uj}) and using the result of
Lemma~\ref{lemma:CDF} we get
\begin{align}
\nonumber  F_l^{(1)}&(x;m) =\Big(F_l(x;m)\Big)^N \\
\nonumber &=\Bigg[1-\exp\bigg[-\frac{x}{\rho\eta_{\min}}-K_m\underset{\leq\frac{P_p}{P_s}\gamma_{\max}(x+1)}
{\underbrace{\ln\left(\frac{P_p}{P_s}\gamma_{\max}x+1\right)}}\bigg]\Bigg]^N\\
\nonumber &\leq \Bigg[1-\exp\bigg[-\frac{x}{\rho\eta_{\min}}-\frac{K_mP_p}{P_s}\gamma_{\max}(x+1)
\bigg]\Bigg]^N\\
\nonumber&=
\Bigg[G\bigg(\frac{x}{\rho\eta_{\min}}+\frac{K_mP_p}{P_s}\;\gamma_{\max}(x+1)\bigg)\Bigg]^N\\
\label{eq:FG1}&=G^{(1)}\bigg(\frac{x}{\rho\eta_{\min}}+\frac{K_mP_p}{P_s}\;\gamma_{\max}(x+1)\bigg).
\end{align}
Now, by using (\ref{eq:FG1}) and by looking at the solutions $x$ and $x'$ of the equations
\begin{align*}
    u & = F_l^{(1)}(x;m),\\
    \mbox{and}\quad u & = G^{(1)}\bigg(\frac{x}{\rho\eta_{\min}}+\frac{K_mP_p}{P_s}\;\gamma_{\max}(x+1)\bigg).
\end{align*}
we find that $x\geq x'$, or equivalently
\begin{equation*}
    \Big(F_l^{(1)}\Big)^{-1}(u;m)\geq\frac{\Big(G^{(1)}\Big)^{-1}(u)- \frac{K_mP_p}{P_s}\gamma_{\max}}{\frac{1}{\rho\eta_{\min}}+\frac{K_mP_p}{P_s}\gamma_{\max}},
\end{equation*}
which after some simple manipulations leads to
\begin{align*}
    \nonumber\log\bigg[\Big(F_l^{(1)}\Big)^{-1}(u;m)+1\bigg]&\geq
    \log\bigg[\rho\eta_{\min}\Big(G^{(1)}\Big)^{-1}(u)+1\bigg]\\
    &-\log\bigg[1+\frac{K_mP_p}{P_s}\gamma_{\max}\rho\eta_{\min}\bigg].
\end{align*}
Therefore, for the lower bound on $R^u_{\max}$ given in
(\ref{eq:RUmax_bounds2_1}) we have
\begin{align*}
\int_0^\infty\log(1+x)\;&dF_l^{(1)}(x;m)\\
&= \int_0^1\log\bigg[1+\Big(F_l^{(1)}\Big)^{-1}(u;m)\bigg]du \\
&\geq
\int_0^1\log\bigg[\rho\eta_{\min}\Big(G^{(1)}\Big)^{-1}(u)+1\bigg]du\\
&-\log\bigg[1+\frac{K_mP_p}{P_s}\gamma_{\max}\rho\eta_{\min}\bigg],
\end{align*}
which is the desired inequality in (\ref{eq:lemma:G2}).

Now, note that $F_u(x;m)\geq G(\frac{x}{\rho\eta_{\max}})$ or
equivalently,
\begin{equation}\label{eq:lemma:G4}
    \Big(F_u^{(1)}\Big)^{-1}(u;m)\leq \rho\eta_{\max}\Big(G^{(1)}\Big)^{-1}(u).
\end{equation}
Therefore,
\begin{align*}
  \int_0^\infty\log(1+x)\;&dF_u^{(1)}(x;m) \\
  &= \int_0^1\log\bigg[1+\Big(F_u^{(1)}\Big)^{-1}(u;m)\bigg]du \\
  &\leq  \int_0^1\log\bigg[1+\rho\eta_{\max}\Big(G^{(1)}\Big)^{-1}(u)\bigg]du\\
  &=\int_0^\infty\log(1+\rho\eta_{\max}x)\;dG^{(1)}(x),
\end{align*}
which establishes the inequality in (\ref{eq:lemma:G1}) and completes
the proof.
\end{proof}

Next, by using the result of the following lemma, we establish the
scaling law of $R^u_{\max}$.

\begin{lemma}
\label{lemma:exp_scaling} For a family of exponentially distributed
random variables of size $N$ and parent distribution $G(x)$ (CDF) and
for any positive real number $a\in\mathbb{R}_+$ we have
\begin{equation*}
    \int_0^\infty\log(1+a x)dG^{(1)}(x)\doteq\log\log N+\log a.
\end{equation*}
\end{lemma}

\begin{proof}
See Appendix \ref{app:lemma:exp_scaling}.
\end{proof}
Now, by recalling the bounds provided in (\ref{eq:RUmax_bounds2_1}) and (\ref{eq:RUmax_bounds2_2}) and taking into account the results of Lemmas \ref{lemma:D}, \ref{lemma:G} and \ref{lemma:exp_scaling} we find the optimal throughput scaling law of cognitive networks.

\begin{theorem}
\label{th:centralized}

In a centralized cognitive network with $N$ secondary transmitter-receiver pairs and $M$ available spectrum bands, by optimal
user-channel assignments, the sum-rate throughput of the network scales
as
\begin{equation*}
    R_{\max}\doteq M\log\log N.
\end{equation*}
\end{theorem}
\begin{proof}
By invoking the results of Lemmas~\ref{lemma:G}~and~\ref{lemma:exp_scaling} on the lower and upper bounds on $R^u_{\max}$ given in (\ref{eq:RUmax_bounds2_1}) and (\ref{eq:RUmax_bounds2_2}) we find
\begin{align*}
    &R^u_{\max}\;\dotgt M\log\log
    N-M\log\bigg[\frac{1}{\rho\eta_{\min}}+\frac{K_mP_p}{P_s}\gamma_{\max}\bigg],\\
    \mbox{and}\quad&R^u_{\max} \dotlt M\log\log N+M\log(\rho\eta_{\max}),
\end{align*}
or equivalently,
\begin{align*}
    \lim_{N\rightarrow\infty}\frac{R^u_{\max}}{M\log\log
    N}&\geq 1-\lim_{N\rightarrow\infty}\frac{\log\bigg[1+\frac{K_mP_p}{P_s} \gamma_{\max}\rho\eta_{\min}\bigg]}{\log\log
    N}\\
    &=1,
\end{align*}
and
\begin{align*}
    \lim_{N\rightarrow\infty}\frac{R^u_{\max}}{M\log\log
    N}\leq \;1+\lim_{N\rightarrow\infty}\frac{\log(\rho\eta_{\max})}{\log\log N}=1,
\end{align*}
which confirms that $R^u_{\max}\doteq M\log\log N$. This result, along
with what stated in Lemma~\ref{lemma:D} concludes that $R^l_{\max}\doteq
R^u_{\max}\doteq M\log\log N$, which establishes the proof of the
theorem.
\end{proof}

So far, we have assumed that there exists a decision-making center
that has full knowledge of  all instantaneous channel realizations, i.e., $\{h^m_{i,j}\}$, $\{\gamma_{i,j}\}$, and $\{\eta_i\}$. Also there is no complexity constraints in order to enable exhausting all the possible user-channel assignments and choosing the one which maximizes the sum-throughput of the network.

The assumptions made in this section, while not being practical, are
useful in shedding light on the sum-throughput limit of such cognitive
networks. The results provided in this section can be exploited as the
benchmark to quantify the efficiency of our distributed algorithm
proposed in the following section.

\section{Distributed Spectrum Allocation}
\label{sec:distributed}

In this section we offer our distributed algorithm, where each user
independently of others, makes decision regarding taking over
transmission on any specific spectrum band. We analyze the achievable
sum-throughput of the cognitive network when this distributed algorithm is utilized and show that it is asymptotically optimal.

\subsection{Distributed Algorithm}
\label{sec:algorithm}

In order to refrain from exhaustively searching for the best user-channel matches, we consider assigning the available spectrum bands to only secondary users with a pre-determined minimum link strength. The distributed algorithm involves two major steps; normalizing the $\sinr$s and comparing it with a given threshold.  The underlying motivation for normalizing the $\sinr$s is to balance fairness among the secondary users, in the sense that they get equal opportunities for accessing the spectrum. It so happens that some transmitter-receiver pairs have a very strong link and some other pair a very weak link. This becomes even more likely when we have a large number of secondary pairs. In such scenarios if spectrum allocation is carried out merely based on the links' strengths, all the strong users will dominate the network and the weak users will hardly have an opportunity for accessing it. So for maintaining fairness, instead of comparing the links' strengths (or $\sinr$s), we compare normalized $\sinr$s. However, it should be noted that the normalization factors have to be designed carefully such that we do not sacrifice achieving the optimal scaling in favor of achieving fairness. In other words, the objective is to achieve the optimal scaling and fairness simultaneously.

In the next step, the normalized $\sinr$s are compared against a pre-determined threshold level and only the users with strong enough links that satisfy the threshold condition will take part in the competition for accessing the spectrum.

Specifically, to each user $n=1,\dots, N$ and channel $m=1,\dots, M$,
we assign a minimum acceptable level of $\sinr$, denoted by
$\lambda(m,n)$, which is defined as follows.

Let $T(x;m,n)$ denote the CDF of $\sinr_{m,n}$ given in
(\ref{eq:sinr}). $\lambda(m,n)$ is set such that we have
\begin{equation}\label{eq:lambda}
    T\Big(\lambda(m,n);m,n\Big)=1-\frac{1}{N}.
\end{equation}
Note that for any given $m$ and $n$, $T(x;m,n)$ is a non-decreasing
function on $[0,+\infty)\times[0,1]$ which ensures that there always
exists a unique solution for $\lambda(m,n)$. Also note that as
$\sinr_{m,n}$ depends only on the incoming channels to the $n^{th}$
secondary receiver on the $m^{th}$ spectrum band, $\lambda(m,n)$ can be computed locally at the $n^{th}$ secondary receiver and does not impose any information exchange between the secondary users. It is assumed that the secondary users are aware of the number of secondary pairs $N$ in the cognitive network.

Now, each user $n$ computes $\sinr_{1,n},\dots,\sinr_{M,n}$, normalizes them via dividing them by $\lambda(1,n)$ $,\dots,\lambda(M,n)$, respectively, and identifies the channel with the largest normalized $\sinr$, and denotes its index by $m^\dag_n$, i.e.,
\begin{equation}\label{eq:m_hat}
    m^\dag_n\dff\arg\max_{m\in\{1,\dots,M\}}\left\{\frac{\sinr_{m,n}}{\lambda(m,n)}\right\}.
\end{equation}
In the next step, the $n^{th}$ user compares $\sinr_{m^\dag_n,n}$
against $\lambda(m^\dag_n,n)$ and if $\sinr_{m^\dag_n,n}\geq
\lambda(m^\dag_n,n)$, then deems itself as a candidate for accessing
the channel indexed by $m^\dag_n$. Based on the definition in (\ref{eq:m_hat}) we define the mutually disjoint sets $\mathcal{H}_m$ for $m=1,\dots,M$, such that $\mathcal{H}_m$ contains the indices of the users deemed as candidates for taking over the $m^{th}$ channel, i.e.,
\begin{equation*}
    \mathcal{H}_m\dff\left\{n\;\Big |\; \arg\max_{m'}\frac{\sinr_{m',n}}{\lambda(m',n)}=m\;\&\;\frac{\sinr_{m,n}|}{\lambda(m,n)}\geq 1 \right\}.
\end{equation*}
Finally, a user with its index in $\mathcal{H}_m$ is \emph{randomly}
opted for utilizing the $m^{th}$ channel. This can be facilitated in a
distributed way via any contention based random media access method,
e.g., Aloha, carrier sense multiple access, etc. As soon as one user
takes a channel, the other users will no longer try to access that
channel. In the following section, we analyze the sum-rate throughput
scaling factor of the proposed algorithm.

\subsection{Sum-Rate Throughput Scaling}
\label{sec:throughput}

We denote the sum-rate throughput by $R_{\rm sum}$ and denote the
throughput of the $m^{th}$ channel by $R_m$. Note that the construction of $\mathcal{H}_m$ guarantees that no single user will be regarded as a candidate for more than one channel and also we have $R_{\rm
sum}=\sum_{m=1}^{M}R_m$. By defining $R_{m\med \mathcal{H}_m}$ as the
throughput achieved for the $m^{th}$ conditioned on having
users with indices in $\mathcal{H}_m$ be candidates for taking over
$B_m$ we have
\begin{equation}\label{eq:Rm1}
    R_m=\sum_{\mathcal{H}_m\subseteq \mathcal{N},\ \mathcal{H}_m\neq \emptyset}R_{m\med
    \mathcal{H}_m}P(\mathcal{H}_m).
\end{equation}
On the other hand, by noting that one member of $\mathcal{H}_m$ will be \emph{randomly} picked for accessing $B_m$ we get
\begin{equation}\label{eq:R_Hm}
    R_{m\med \mathcal{H}_m}=\frac{1}{|\mathcal{H}_m|}\sum_{i\in\mathcal{H}_m}
    \bbe\bigg[\log\big(1+\sinr_{m,i}\Big)\;\Big
    |\;\mathcal{H}_m\bigg].
\end{equation}
From (\ref{eq:Rm1}) and (\ref{eq:R_Hm}) for any  $\mathcal{H}_m\neq\emptyset$ we get
\begin{equation}\label{eq:Rm2}
    R_m=\sum_{\mathcal{H}_m\subseteq \mathcal{N}}
    \frac{P(\mathcal{H}_m)}{|\mathcal{H}_m|}\sum_{i\in\mathcal{H}_m} \bbe\bigg[\log\big(1+\sinr_{m,i}\Big)\;\Big
    |\;\mathcal{H}_m\bigg].
\end{equation}
As shown in Appendix \ref{app:lower} we have
\begin{align}
    \nonumber \sum_{i\in\mathcal{H}_m}\bbe\bigg[\log\Big(1+&\sinr_{m,i}\Big)\;\bigg|\;
    \mathcal{H}_m\bigg]\\
    \label{eq:Rm_lower} \geq &
    \sum_{i=1}^{|\mathcal{H}_m|}\bbe\bigg[\log\Big(1+\sinr_m^{(i)}\Big)\;\bigg].
\end{align}
Therefore, (\ref{eq:Rm2}) and (\ref{eq:Rm_lower}) together give rise
to the following lower bound on $R_m$
\begin{align}\label{eq:Rm3}
    \nonumber R_m &\geq\sum_{\mathcal{H}_m\subseteq \mathcal{N},\
    \mathcal{H}_m\neq \emptyset}\frac{P(\mathcal{H}_m)}{|\mathcal{H}_m|}
    \sum_{i=1}^{|\mathcal{H}_m|}\bbe\bigg[\log\Big(1+\sinr_m^{(i)}\Big)\;\bigg]\\
    \nonumber & = \sum_{n=1}^N\sum_{|{\cal H}_m|=n}\frac{P(\mathcal{H}_m)}{|\mathcal{H}_m|}
    \sum_{i=1}^{|\mathcal{H}_m|}\bbe\bigg[\log\Big(1+\sinr_m^{(i)}\Big)\;\bigg]\\
    \nonumber & = \sum_{n=1}^N \sum_{i=1}^{n}\bbe\bigg[\log\Big(1+\sinr_m^{(i)}\Big)\;\bigg]\sum_{|{\cal H}_m|=n}\frac{P(\mathcal{H}_m)}{n}\\
    \nonumber & = \sum_{n=1}^N \sum_{i=1}^{n}\bbe\bigg[\log\Big(1+\sinr_m^{(i)}\Big)\;\bigg]\frac{P(|\mathcal{H}_m|=n)}{n}\\ & = \sum_{i=1}^{N}\bbe\bigg[\log\Big(1+\sinr_m^{(i)}\Big)\;\bigg]\sum_{n=i}^N \frac{P(|\mathcal{H}_m|=n)}{n}.
\end{align}
By further defining
\begin{equation}\label{eq:q1}
    Q^m_i\dff \sum_{n=i}^N\frac{1}{n}P\Big(|\mathcal{H}_m|=n\Big),
\end{equation}
and
\begin{equation}\label{eq:q2}
    R_m^l\dff\sum_{i=1}^N Q^m_i \bbe\bigg[\log\Big(1+\sinr_m^{(i)}\Big)\;\bigg],
\end{equation}
we can re-write (\ref{eq:Rm3}) as $R_m\geq R^l_m$. If we also
define $Q_0=P\Big(|\mathcal{H}_m|=0\Big)$ we get
\begin{align*}
    \sum_{i=0}^NQ^m_i&=P\Big(|\mathcal{H}_m|=0\Big)+\sum_{i=1}^N
    \sum_{n=i}^N\frac{1}{n}P\Big(|\mathcal{H}_m|=n\Big)\\
    &=
    \sum_{n=0}^NP\Big(|\mathcal{H}_m|=n\Big)=1,
\end{align*}
which suggests that $\{Q^m_i\}_{i=0}^N$ is a valid probability mass
function (pmf). In the sequel, we concentrate on finding the scaling
behavior of $R^l_m$. By using the definitions of $\mathcal{S}_l(m)$ and
$\mathcal{S}_u(m)$ and exploiting Lemma~\ref{lemma:order}, from (\ref{eq:q2}) we have
\begin{align}
    \label{eq:Rl_bounds1_1}R^l_m\geq \sum_{i=1}^N Q^m_i \bbe\bigg[\log\Big(1+\mathcal{S}_l^{(i)}(m)\Big)\;\bigg],\\
    \label{eq:Rl_bounds1_2}\mbox{and}\;\;\;R^l_m\leq \sum_{i=1}^N Q^m_i
    \bbe\bigg[\log\Big(1+\mathcal{S}_u^{(i)}(m)\Big)\;\bigg].
\end{align}
By recalling that the CDFs of $\mathcal{S}^{(i)}_l(m)$ and
$\mathcal{S}^{(i)}_u(m)$ are $F^{(i)}_l(x;m)$ and $F^{(i)}_u(x;m)$
provided in (\ref{eq:CDF:Lj}) and (\ref{eq:CDF:Uj}), respectively, (\ref{eq:Rl_bounds1_1}) and (\ref{eq:Rl_bounds1_2}) can
be stated as
\begin{align}
    \label{eq:Rl_bounds2_1} R^l_m\geq \sum_{i=1}^NQ^m_i\int_0^1\log(1+x)\;dF^{(i)}_l(x;m)\\ \label{eq:Rl_bounds2_2} \mbox{and}\;\;\;R^l_m\leq \sum_{i=1}^NQ^m_i\int_0^1\log(1+x)\;dF^{(i)}_u(x;m),
\end{align}
Next, for the given set of $\{Q^m_i\}$ we define
\begin{align}
    \label{eq:CDFN1}F^N_l(x;m)&\dff\sum_{i=1}^NQ^m_iF^{(i)}_l(x;m),\\
    \label{eq:CDFN2}F^N_u(x;m)&\dff\sum_{i=1}^NQ^m_iF^{(i)}_u(x;m),\\
    \label{eq:CDFN3}\mbox{and}\;\;\;G^N(x)&\dff\sum_{i=1}^NQ^m_iG^{(i)}(x).
\end{align}
Since $\{Q^m_i\}$ is a valid pmf, $F^N_l(x;m)$, $F^N_u(x;m)$, and
$G^N(x)$ can be cast as valid CDFs. Therefore, (\ref{eq:Rl_bounds2_1}) and (\ref{eq:Rl_bounds2_2}) give rise to
\begin{equation}\label{eq:Rl_bounds2}
    \int_0^1\log(1+x)\;dF^N_l(x;m)\leq R^l_m \leq\int_0^1\log(1+x)\;dF^N_u(x;m).
\end{equation}
The two subsequent lemmas are key in finding how $R^l_m$ scales with
increasing $N$.

\begin{lemma}
\label{lemma:f}
For a real variable $x\in[0,1]$ and integer variables $N$ and $i$,
$0\leq i\leq N-1$, the function
\begin{equation*}
    f(x,i)\dff\sum_{j=0}^i{N\choose j}x^{N-j}(1-x)^j,
\end{equation*}
is increasing in $x$.
\end{lemma}
\begin{proof}
See Appendix \ref{app:lemma:f}.
\end{proof}

\begin{lemma}
\label{lemma:G2} For the distributions $F_l^N(x)$, $F_u^N(x)$ and
$G^N(x)$ we have
\begin{align}
  \label{eq:lemma:G2_1}\int_0^\infty\log(1+x)\;dF_u^N(x;m) &\leq \int_0^\infty\log(1+\rho\eta_{\max} x)\;dG^N(x),
\end{align}
and
\begin{align}
  \nonumber \int_0^\infty\log(1+x)\;dF^N_l(x;m) &\geq
  \int_0^\infty\log(1+\rho\eta_{\min}x)\;dG^N(x)\\
  \label{eq:lemma:G2_2}&-
  \log\bigg[1+\frac{K_mP_p}{P_s}\gamma_{\max}\rho\eta_{\min}\bigg].
\end{align}
\end{lemma}
\begin{proof}
By using the definition of $f(x,j)$ provided in Lemma~\ref{lemma:f} and
recalling (\ref{eq:CDF:Lj})-(\ref{eq:CDF:Uj}) we have
\begin{align*}
    F_l^{(i)}(x;m)&=f\Big(F_l(x;m),i-1\Big),\\ F_u^{(i)}(x;m)&=f\Big(F_u(x;m),i-1\Big),\\
    \mbox{and}\;\;\;G^{(i)}(x)&=f\Big(G(x),i-1\Big).
\end{align*}
By following the same lines as in (\ref{eq:FG1}) we also have
\begin{eqnarray*}
F_l(x;m) \leq
G\bigg(\frac{x}{\rho\eta_{\min}}+\frac{K_mP_p}{P_s}\;\gamma_{\max}(x+1)\bigg),
\end{eqnarray*}
and consequently by applying Lemma~\ref{lemma:f} and using the
definition in (\ref{eq:CDFN1})-(\ref{eq:CDFN3}) we have
\begin{align}
    \nonumber F^N_l&(x;m)\\
    \nonumber &=\sum_{i=1}^NQ^m_iF^{(i)}_l(x;m)=\sum_{i=1}^NQ^m_if\Big(F_l(x;m),i-1\Big)\\
    \nonumber   &\leq
    \sum_{i=1}^NQ^m_if\Bigg(G\bigg(\frac{x}{\rho\eta_{\min}}+
    \frac{K_mP_p}{P_s}\;\gamma_{\max}(x+1)\bigg),i-1\Bigg)\\
    \nonumber&=\sum_{i=1}^NQ^m_iG^{(i)}\bigg(\frac{x}{\rho\eta_{\min}}+\frac{K_mP_p}{P_s}\;\gamma_{\max}(x+1)\bigg)\\
    \nonumber &=G^N\bigg(\frac{x}{\rho\eta_{\min}}+
    \frac{K_mP_p}{P_s}\;\gamma_{\max}(x+1)\bigg)
\end{align}
By following a similar approach as in Lemma~\ref{lemma:G}, and
(\ref{eq:FG1})-(\ref{eq:lemma:G4}) the inequality in (\ref{eq:lemma:G2_2}) can be established. Proof of (\ref{eq:lemma:G2_1}) follows a similar line of argument.
\end{proof}

\begin{lemma}
\label{lemma:exp_scaling2}

For a family of exponentially distributed random variables of size $N$
and parent distribution $G(x)$ (CDF) and for any set of
$\{Q^m_i\}_{i=1}^N$ such that $\sum_{i=0}^NQ^m_i=1$, if the condition
\begin{equation}\label{eq:Q1}
    \lim_{N\rightarrow\infty}\frac{\sum_{i=1}^NiQ^m_i}{N}=0,
\end{equation}
is satisfied, then for any positive real number $a\in\mathbb{R}_+$ we have
\begin{equation*}
    \int_0^\infty\log(1+a x)dG^N(x)\doteq\log\log N+\log a.
\end{equation*}
\end{lemma}

\begin{proof}
See Appendix \ref{app:lemma:exp_scaling2}.
\end{proof}

By using the results of the Lemmas \ref{lemma:G2} and
\ref{lemma:exp_scaling2} we offer the main result of the
distributed algorithm in the following theorem

\begin{theorem}
\label{th:distributed}

The sum-rate throughput of the cognitive network by exploiting the
proposed distributed algorithm scales as
\begin{equation*}
    R_{\rm sum}\doteq M\log\log N.
\end{equation*}

\end{theorem}
\begin{proof}
We start by demonstrating that the set $\{Q^m_i\}$ as defined in
(\ref{eq:q1}) fulfils the condition (\ref{eq:Q1}) of Lemma
\ref{lemma:exp_scaling2}. From (\ref{eq:q1}) we have
\begin{align}\label{eq:Q2}
    \nonumber\sum_{i=1}^NiQ^m_i&=\sum_{i=1}^Ni\sum_{n=i}^N\frac{1}{n}P\Big(|\mathcal{H}_m|=n\Big)\\
    \nonumber &=
    \sum_{n=1}^N\frac{1}{n}P\Big(|\mathcal{H}_m|=n\Big)\sum_{i=1}^ni\\
    \nonumber &=
    \sum_{n=1}^N\frac{n+1}{2}P\Big(|\mathcal{H}_m|=n\Big)\\
    \nonumber&=
    \frac{1}{2}\sum_{n=1}^NnP\Big(|\mathcal{H}_m|=n\Big)\\
    &+
    \frac{1}{2}\underset{=1}{\underbrace{\sum_{n=0}^NP\Big(|\mathcal{H}_m|=n\Big)}}-\frac{1}{2}P\Big(|\mathcal{H}_m|=0\Big).
\end{align}
Note that $|\mathcal{H}_m|$ has \emph{compound} binomial distribution
with parameters $\{p(m,n)\}_{n=1}^N$ \cite{Johnson}, where $p(m,n)$
denotes the probability that the $m^{th}$ channel is allocated to the
$n^{th}$ user. Therefore, according to the properties of compound
binomial distributions we have \cite{Johnson}
\begin{equation}\label{eq:Q3}
    \sum_{n=1}^NnP\Big(|\mathcal{H}_m|=n\Big)=\bbe\left[|\mathcal{H}_m|\right]=\sum_{n=1}^Np(m,n).
\end{equation}
From (\ref{eq:Q2}) and (\ref{eq:Q3}) we get
\begin{equation*}
    \sum_{i=1}^NiQ^m_i\leq \frac{1}{2}\left(\sum_{n=1}^Np(m,n)+1\right).
\end{equation*}
On the other hand, the probability that any specific user $n$ can be a
candidate for taking over \emph{any} of the $M$ channels is
\begin{align*}
    \nonumber \omega(n)&\dff P\left(\max_m\left\{\frac{\sinr_{m,n}}{\lambda(m,n)}\right\}\geq1\right)\\
    &=
    1-\prod_{m=1}^{M}P\Big(\sinr_{m,n}\leq \lambda(m,n)\Big)\\
    \nonumber &=1-\prod_{m=1}^{M}\underset{=1-1/N}{\underbrace{T\Big(\lambda(m,n);m,n\Big)}}
    =1-\left(1-\frac{1}{N}\right)^{M}.
\end{align*}
Therefore, since $\sum_{m=1}^Mp(m,n)=\omega(n)$, for all $m,n$ we have
$p(m,n)\leq \omega(n)$. Hence,
\begin{equation*}
    \sum_{i=1}^NiQ^m_i\leq \frac{1}{2}\left(N\omega(n)+1\right).
\end{equation*}
On the other hand,
\begin{equation*}
    \lim_{N\rightarrow\infty}N\omega(n)=
    \lim_{N\rightarrow\infty}\frac{1-\left(1-\frac{1}{N}\right)^{M}}{\frac{1}{N}}=M.
\end{equation*}
Therefore, $\sum_{i=1}^NiQ^m_i\leq \frac{1}{2}(M+1)$ and the set
$\{Q^m_i\}$ satisfies the condition in Lemma \ref{lemma:exp_scaling2}.
Therefore, Lemmas \ref{lemma:G2} and \ref{lemma:exp_scaling2} together
establish the following
\begin{align*}
  \int_0^\infty\log(1+x)\;dF_u^N(x;m) &\dotlt \log\log
  N+\log(\rho\eta_{\max}),
  \\
  \nonumber\mbox{and}\;\;\int_0^\infty\log(1+x)\;dF^N_l(x;m) &\dotgt
  \log\log N\\
  - \log\bigg[\frac{1}{\rho\eta_{\min}}&+\frac{K_mP_p}{P_s}\gamma_{\max}\bigg].
\end{align*}
The two inequalities above, in conjunction with (\ref{eq:Rl_bounds2}) and noting that $R_{\rm sum}=\sum_{m=1}^{M}R_m$ provide
\begin{align*}
    \nonumber M\log\log
    N-&M\log\bigg[\frac{1}{\rho\eta_{\min}}+\frac{K_mP_p}{P_s}\gamma_{\max}\bigg]\\
    \nonumber&\dotlt \; R_{\rm sum}\\
    &\dotlt\;  M\log\log N+M\log(\rho\eta_{\max}),
\end{align*}
which concludes the desired result.
\end{proof}

\subsection{Simulation Results}
\begin{figure}[t]
  \centering
  \includegraphics[width=3.7 in]{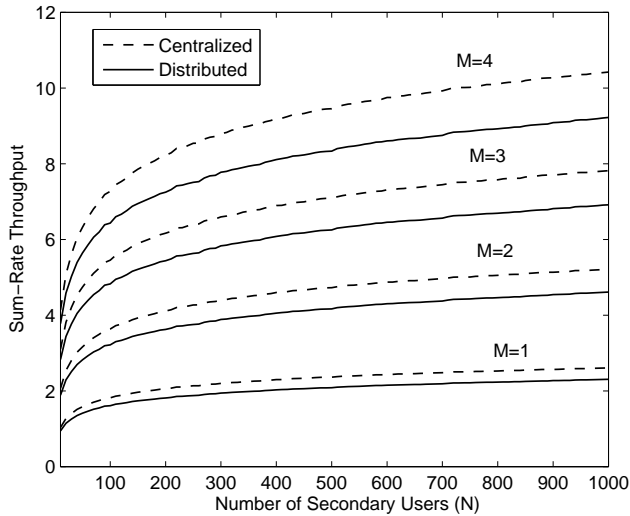}\\
  \caption{Sum-rate throughput versus the number of secondary users for
  $M=1,\dots, 4$, $\rho=10$ dB, and the number of the primary users
  $K_m=4$.}\label{fig:sumrate}
\end{figure}

The simulation results in Fig.~\ref{fig:sumrate} demonstrates the
sum-rate throughput achieved under the centralized setup given in
(\ref{eq:Rmax}) and the distributed setup given in Section
\ref{sec:algorithm}. We consider a primary network consisting of 4
users and look at the throughput scaling for the cases that there exist
$M=1,\dots,4$ available spectrum bands to be utilized by the secondary
users. We set all path-loss terms $\{\eta_i\}$ and $\{\gamma_{i,j}\}$
equal to 1 and find the sum-rate throughput as the number of secondary
users increases. As shown in Fig.~\ref{fig:sumrate}, as the number of
secondary users increases, the sum-rate throughput achieved by the
centralized and distributed schemes exhibit the same scaling factor.
Note that what Theorems 1 and 2 convey is that the ratio of $R_{\max}$ and $R_{\rm sum}$ in the centralized and distributed setups, respectively, approaches to 1 as $N\rightarrow\infty$, i.e.,
\begin{equation*}
   \lim_{N\rightarrow\infty}\frac{R_{\max}}{R_{\rm sum}}=1
\end{equation*}
which does not necessarily mean that $R_{\max}$ and $R_{\rm sum}$ have to coincide. As a matter of fact, as observed in Fig. 1, there is a gap between $R_{\max}$ and $R_{\rm sum}$, that according to the results of Theorems 1 and 2 must be diminishing with respect to $R_{\max}$ and $R_{\rm sum}$ such that we obtain the asymptotic equality $R_{\max}\;\doteq\;R_{\rm sum}\;\doteq\; M\log\log N$. This gap accounts for the cost incurred for enabling distributed processing in the distributed spectrum access algorithm.

The throughput achieved under the distributed setup uniformly is less
than that of the centralized setup. This is justified by recalling that
the centralized scheme finds the best secondary user for each available
spectrum band, whereas the distributed network finds all the secondary
users whose quality of communication on a specific channel satisfies a
constraint ($\lambda(m,n)$) and among all such secondary user one is
randomly selected to access the spectrum band. This, not necessarily
guarantees finding the best user for each available spectrum band and as a result leads to some degradation in the sum-rate throughput.

\begin{figure}[t]
  \centering
  \includegraphics[width=3.7 in]{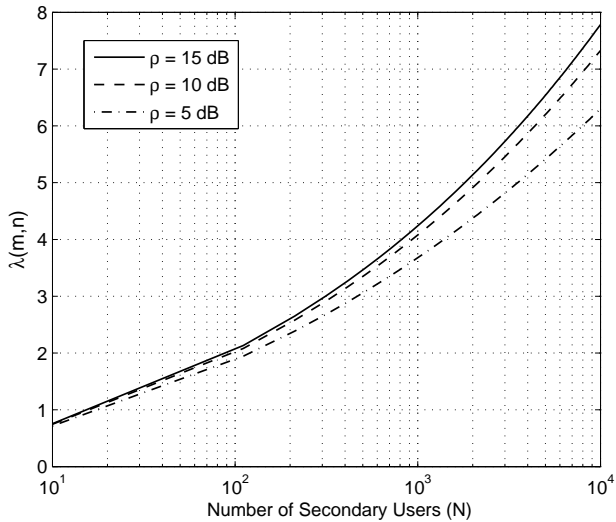}\\
  \caption{$\lambda(m,n)$ versus the number of secondary users for $M=4$
  available spectrum bands and $K_m=4$ primary users.}\label{fig:threshold1}
\end{figure}

Finding the metric $\lambda(m,n)$ as defined in (\ref{eq:lambda}) is
the heart of the distributed spectrum allocation algorithm. As it is
not mathematically tractable to formulate the CDF of $\sinr_{m,n}$
(note that $F_l(x;m)$ and $F_u(x;m)$ are only the CDFs of the lower and
upper bounds on $\sinr_{m,n}$), we are not able to find a closed form
expression for $\lambda(m,n)$. However, by solving
(\ref{eq:lambda}) numerically, we provide the following two figures which are helpful in shedding light on how $\lambda(m,n)$ varies with other network parameter, i.e., primary and cognitive network sizes as well as the number of available spectrum bands.

Figure \ref{fig:threshold1} demonstrates the dependence of
$\lambda(m,n)$ on the transmission $\snr$ denoted by $\rho$. It is seen
that $\lambda(m,n)$ monotonically increases with $\rho$. Intuitively, as $\rho$ increases, the users are expected to have more reliable communication and as a result the algorithm will impose more stringent conditions on the secondary users for considering themselves as a candidate for accessing any specific spectrum band. More stringent conditions will translate to having higher values of $\lambda(m,n)$ such that the condition in (\ref{eq:lambda}) is satisfied.

The numerical evaluations provided in Fig. \ref{fig:threshold2} show that $\lambda(m,n)$ increases as the size of the primary network decreases. Again as in Fig. \ref{fig:threshold1}, this is justified by noting that smaller number of primary users leads to less interference from the primary network to the cognitive network and thereof, more reliable secondary links. Thus, decreasing the primary network size again requires more stringent conditions to be satisfied for a secondary user to be deemed as a candidate for taking over a spectrum band, which in turn results in an increase in $\lambda(m,n)$. It is noteworthy that the choices of the thresholds given in (\ref{eq:lambda}) have been {\em heuristic} choices that satisfy all the desired properties (optimal scaling as well as fairness and limited information exchange as discussed in Section \ref{sec:properties}). Nevertheless, we cannot prove that these are the only choices of the thresholds and it {\em might} be possible to find some other threshold settings that satisfy all these conditions and yet do not depend on the size of the primary network. Hence, while the scaling of the sum-rate throughput does not depend on the size of the primary network, the choices of the thresholds for achieving this scaling in the distributed algorithm do depend on the size of the primary network.

\begin{figure}[t]
  \centering
  \includegraphics[width=3.7 in]{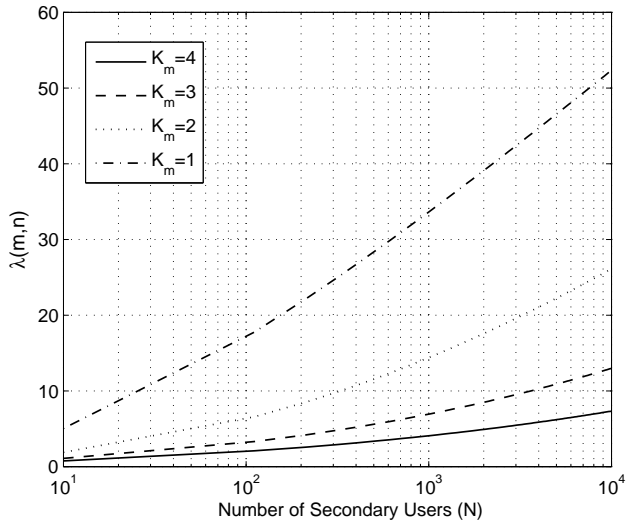}\\
  \caption{$\lambda(m,n)$ versus the number of secondary users for different
  sizes for the primary network $K_m=1,\dots, 4$ and $\rho=10$ dB.}\label{fig:threshold2}
\end{figure}

\section{Information Exchange and Fairness}
\label{sec:properties}

\subsection{Information Exchange}
\label{sec:information}

In the distributed algorithm we assume that the $n^{th}$ secondary
\emph{receiver} measures $\{\sinr_{m,n}\}_{m=1}^{M}$ corresponding to
different spectrum bands, selects the largest one, whose index is
denoted by $m_n^\dag$, and compares it against a pre-determined quality metric $\lambda(m,n)$. If $\sinr_{m^\dag_n,n}\geq \lambda(m,n)$, this
secondary receiver should notify its designated secondary transmitter
to participate in a contention-based competition for taking over
channel $m^\dag_n$. Such notification requires transmitting $\log M$
\emph{information bits} from the secondary receiver to its respective
secondary transmitter.

Although not all of the secondary pairs will be involved in such
information exchange, it is imperative to analyze the aggregate amount
of such information for large networks ($N\rightarrow\infty$). In the
following theorem we demonstrate that for the choice of $\lambda(m,n)$
provided in (\ref{eq:lambda}), the asymptotic average amount of information exchange is a constant independent of $N$ and therefore does not harm the sum-rate throughput of the cognitive network.

\begin{theorem}
\label{th:information}

In the cognitive network with distributed spectrum access, when
$\lambda(m,n)$ satisfies
\begin{equation*}
    T\Big(\lambda(m,n);m,n\Big)=1-\frac{1}{N},
\end{equation*}
the average aggregate amount of information exchange between secondary
transmitter-receiver pairs is asymptotically equal to $M\log M$.
\end{theorem}
\begin{proof}
As stated earlier, the probability that the a user satisfies the
$\lambda(m,n)$ constraint is $\omega(n)=1-(1-1/N)^{M}$. Therefore, the
average aggregate amount of information exchange, denoted by $R_{\rm
ie}$, is
\begin{align}
  \nonumber R_{\rm ie} &= \lim_{N\rightarrow\infty}N\omega(n)\log M\\
  \label{eq:IE1} &=\log M  \lim_{N\rightarrow\infty}\frac{1-(1-\frac{1}{N})^{M}}{\frac{1}N}\\
\label{eq:IE2}&= \log M
\lim_{N\rightarrow\infty}\frac{-\frac{1}{N^2}M(1-\frac{1}{N})^{M-1}}{-\frac{1}{N^2}}=
M\log M,
\end{align}
where for the transition from (\ref{eq:IE1}) to (\ref{eq:IE2}) we have
used the L'Hopital's rule.
\end{proof}

\subsection{Fairness} \label{sec:fairnes}

In general, opportunistic user selections might lead to the situation that the network be dominated by the secondary pairs with
their receiver far from the primary users so that see less amount of
interference from them, or by those pairs where the transmitter and the
receiver are closely located and enjoy a good communication channel.

In despite of these facts, we show that in our network, by appropriately choosing $\{\lambda_{m,n}\}$ we can provide equiprobable opportunity for all users to access the available spectrum bands. This can be made possible by enforcing more stringent conditions (higher $\lambda(m,n)$) for the users benefitting from smaller path-loss and shadowing effects. In the following theorem we show that by the choice of $\lambda(m,n)$ provided in (\ref{eq:lambda}) all the users have the equal opportunities for accessing a channel.

\begin{theorem}
\label{th:fairness} In the cognitive network with distributed spectrum
allocation, when $\lambda(m,n)$ satisfies
\begin{equation*}
    T\Big(\lambda(m,n);m,n\Big)=1-\frac{1}{N},
\end{equation*}
all the users have the same probability for being a allocated a
channel.
\end{theorem}
\begin{proof}
As shown earlier, the probability that user $n$ satisfies the $\sinr$
constraint $\lambda(m,n)$ is $\omega(n)=1-(1-1/N)^{M}$, which is the same for all users.
\end{proof}

\section{Discussions}
\label{sec:discussions}

\subsection{Impact on the Primary Network}
In cognitive networks with {\em underlaid} spectrum access,
the secondary and primary users may coexist simultaneously. Therefore, in order to protect the primary users, the secondary users must adjust their transmission power such that they operate within the tolerable noise level of the primary users and thereof do not harm the communication of the primary users. Hence, it is imperative to investigate whether such power adjustments affect the achievable throughput scaling in the centralized and distributed setups.

According to theorems \ref{th:centralized} and \ref{th:distributed}, the sum-rate throughput of the cognitive network scales as $M\log\log N$ which does not depend on the $\snr$ or the transmission power of the secondary users. Therefore, irrespective of the transmission policy and power control mechanism (i.e., for any arbitrary $\snr$ or transmission power), the secondary users achieve the scaling law of $M\log\log N$. Hence, deploying any power management mechanism of interest along with the proposed spectrum access algorithms, does not harm the optimal scaling.

\subsection{Distributed Algorithm}
For implementing the distributed spectrum access protocol there are two major steps involved. First the random selection of a user out of the candidates for taking over a specific spectrum band. For randomly selecting a user out of the set of users in ${\cal H}_m$ to access the $m^{th}$ spectrum band, one distributed approach is to equip all the users with a backoff timers. Then when a user learns that it is a candidate for accessing the $m^{th}$ spectrum band with run the backoff timer with an initial random value. The first cognitive pair whose backoff timer goes off will take over the channel and with a beacon message can notify it to the rest of the network.

Secondly, the distributed algorithm requires some secondary receivers transmit $\log M$ information bits to their respective transmitters. Transmitting $\log M$ information bits requires a very low rate communication. An appropriate approach for such communication rate is to deploy ultra-wide band (UWB) communication between a secondary transmitter and receiver pair. This will allow the secondary users to communicate the low-rate information bits well below the noise level of the primary users. It is noteworthy that the cognitive radios are often assumed to be equipped with wideband filters which enable them to transmit and receive in a wide range of frequency spectrum. This feature of the cognitive radios provides an appropriate context for implementing UWB communication.

\subsection{Networks of Limited Size}
Practical networks do not have large enough number of users to fully capture the multiuser diversity gain (double-logarithmic growth of capacity with the number of users). Therefore, due to such degradation in multiuser diversity gain, the network cannot support the throughput expected in theory. Therefore, for practical networks we have only {\em upper bounds} on the actual sum-rate throughputs. Knowing such upper bounds help to find in an insight about what to expect from the cognitive networks at the design state.

Although the result hold analytically for only $N\rightarrow\infty$, from the simulation results in Fig. 1 we observe that as low as $N=50$ (which is around the point that we start to observe steady increase in the throughput) secondary users are enough to start observing the multiuser diversity gain. This is not far from the size of practical networks.
\section{Conclusions}
\label{sec:conclusion}

In this paper we investigated the multiuser diversity gain in cognitive
networks. We first obtained the optimal gain achieved in a network with
a central authority and show that the gain achieved in such cognitive
networks is similar to that of the interference-free networks, i.e.,
the network throughput scales double-logarithmically with the number of users. Then we proposed a distributed spectrum access scheme which is
proven to achieve the optimal throughput scaling factor. This scheme
imposes the exchange of $\log M$ information bits per
transmitter-receiver cognitive pair for some pairs, and no information
exchange for the others. The other specification of the distributed
algorithm are that the network-wide average aggregate amount of
information bits it requires is asymptotically equal to $M\log M$, and
it ensures fairness among the secondary users.

\appendix

\section{Proof of Lemma \ref{lemma:D}} \label{app:lemma:D}

We equivalently show that $\lim_{N\rightarrow\infty}P(\mathcal{D})=1$. An intuitive justification is that if we put the $\sinr$s in an $M\times N$ array, and locate the maximum element of each row, event ${\cal D}$ occurs when no two such maximum are located in the same column. Therefore, as the number of the columns increases, in the asymptote of very large values of $N$, the even $\cal D$ must occur with probability 1.

For set $C\subseteq\mathcal{M}\dff\{1,\dots,M\}$ such that $|C|\geq 2$ we define $P(C)$ as the probability that the spectrum bands with indices in
$C$ have the same most favorable user. Therefore, we get
\begin{eqnarray}
    \nonumber P(C)&\dff& P\left\{\forall \;m,m'\in
    C,\;n^*_{m}=n^*_{m'}\right\}\\
    \nonumber &=&\sum_{n=1}^NP\left\{\forall \;m\in
    C,\;n^*_{m}=n\;\Big|\; n^*_m=n\right\}\\
    \nonumber && \hspace{.5in}\times{P(n^*_m=n)}\\
    \nonumber &=&\sum_{n=1}^N\prod_{m\in C}P\Big(\sinr_{m,n}\geq\sinr_{m,n'},\;\forall n'\neq
    n\Big)\\
    \label{eq:lemma:D_proof1} && \hspace{.5in}\times{P(n^*_m=n)}\\
    \nonumber &=&\sum_{n=1}^N\prod_{m\in
    C}\prod_{n\neq n'}\underset{\dff\;q(m,n,n')}{\underbrace{P\Big(\sinr_{m,n}\geq\sinr_{m,n'}\Big)}}\\
    \nonumber&&\hspace{.5in}\times{P(n^*_m=n)}\\
    \label{eq:lemma:D_proof2}
    &\leq& N(q_{\max})^{|C|(N-1)},
\end{eqnarray}
where $q_{\max}=\max_{m,n,n'} q(m,n,n')$ and (\ref{eq:lemma:D_proof1})
and (\ref{eq:lemma:D_proof2}) hold due to the statistical independence of
the elements in $\{\sinr_{m,n}\}$, and . Therefore we have
\begin{eqnarray*}
  P(\mathcal{D}) &=& 1-\sum_{C\subseteq\mathcal{M}, |C|\geq 2}P(C) \\
  &=&1-\sum_{m=2}^{M}\sum_{C\subseteq\mathcal{M},\;{|C|=m}}P(C)\\
  & \geq &1-\sum_{m=2}^{M}\underset{\rightarrow 0\;\mbox{as}\;N\rightarrow\infty}{\underbrace{{M\choose m
  }N(q_{\max})^{m(N-1)}}}\\
  &\doteq&1,
\end{eqnarray*}
which completes the proof. Note that $q_{\max}$ is a function of $\{\eta_i\}$ and $\{\lambda_{i,j}\}$ and does not depend on $N$.

\section{Proof of Lemma \ref{lemma:order}} \label{app:lemma:order}

We first show that for any $i=1,\dots, N$,
$\mathcal{S}_l^{(i)}(m)\leq \sinr^{(i)}_m$. For $i=1$ we have
\begin{equation*}
    \sinr^{(1)}_m=\max_n\sinr_{m,n}\geq\max_nS_l(m,n)=\mathcal{S}_l^{(1)}(m).
\end{equation*}
For any $i=2,\dots,N$, from the definition of $\sinr^{(i)}_m$ and
$\mathcal{S}^{(i)}_m$ it can be deduced that each of the $(N-i+1)$
terms $\sinr^{(i)}_m, \dots, \sinr^{(N)}_m$ is greater than one
corresponding element in the set $\mathcal{S}_l(m)$. Therefore, there
cannot be more than $(i-1)$ elements in $\mathcal{S}_l(m)$ which are
all greater than $\sinr^{(i)}_m, \dots, \sinr^{(N)}_m$.

Now, if $\mathcal{S}_l^{(i)}(m)> \mathcal\sinr^{(i)}_m$, then all the
$i$ terms $\mathcal{S}_l^{(1)}(m)$, $\mathcal{S}_l^{(2)}(m), \dots,
\mathcal{S}_l^{(i)}(m)$ should be greater than all the $(N-i+1)$ terms
$\sinr^{(i)}_m,\dots, \sinr^{(N)}_m$. Therefore, we have found $i$
elements in $\mathcal{S}_l(m)$ that are all greater than
$\sinr^{(i)}_m,\dots, \sinr^{(N)}_m$ and this contradicts with what we
found earlier. Hence, we should have $\mathcal{S}_l^{(i)}(m)\leq
\mathcal\sinr^{(i)}_m$.

By following the same lines, we can show that also for $i=1,\dots, N$,
we always have $\sinr_m^{(i)}\leq\mathcal{S}_u^{(i)}(m)$, which
concludes the proof of the lemma.

\section{Proof of Lemma \ref{lemma:CDF}} \label{app:lemma:CDF}

Let $Y\dff|g^m_n|^2$, which has exponential distribution with unit
variance. Also define $Z\dff\sum_{j=1}^{K_m}|h^m_{n,j}|^2$ which is the
the summation of $K_m$ independent exponentially distributed random variables each with unit variance, and thereof has a ${\rm Gamma}(K_m,1)$ distribution. By denoting the probability density functions (PDF) of $Z$ and $Y$ by
\begin{eqnarray*}
  f_Y(y) &=& e^{-z}, \\
  \mbox{and}\;\;\;f_Z(z) &=&  \frac{z^{K_m-1}\;e^{-z}}{(K_m-1)!},
\end{eqnarray*}
the PDF of $S_l(m,n)=\frac{Y}{1/\rho\eta_{\min}+P_p/P_s\gamma_{\max}Z}$, denoted by $f_S(x)$ is
\begin{align*}
  f_S(x) &= \int_0^{\infty}f_{S\med Z}(x\med z)f_Z(z)dz \\
  &= \int_0^{\infty}\left(\frac{1}{\rho\eta_{\min}}+
  \frac{P_p\gamma_{\max}}{P_s}z\right)e^{\left(-\frac{x}{\rho\eta_{\min}}-
  \frac{P_p\gamma_{\max}}{P_s}zx\right)}dz\\
  &\hspace{1in}\times\frac{z^{K_m-1}\;e^{-z}}{(K_m-1)!}\;dz\\
  &=\frac{e^{-x/\rho\eta_{\min}}}{\left(\frac{P_p}{P_s}\gamma_{\max}x+1\right)^{K_m+1}}\\
  &\times
\left[\frac{\frac{P_p}{P_s}\gamma_{\max}x+1}{\rho\eta_{\min}}+K_m\frac{P_p}{P_s}\gamma_{\max}\right],
\end{align*}
where the last step holds as $\int_0^\infty e^{-u}u^M=M!$. Therefore The CDF is
\begin{equation*}
    F_l(x;m)=
    1-\frac{e^{-x/\rho\eta_{\min}}}{\left(\frac{P_p}{P_s}\gamma_{\max}x+1\right)^{K_m}}.
\end{equation*}
$F_u(x;m)$ can be found by following the same lines.

\section{Proof of Lemma \ref{lemma:exp_scaling}}
\label{app:lemma:exp_scaling}

We start by citing the following theorem.
\begin{theorem}\label{th:limit}
\emph{\cite[Theorem 4]{Sanayei:WCOM07}} Let $\{X_n\}_{n=1}^N$ be a family
of positive random variables with finite mean $\mu_N$ and variance
$\sigma^2_N$, also $\mu_N\rightarrow\infty$ and
$\frac{\sigma_N}{\mu_N}\rightarrow 0$ as $N\rightarrow\infty$. Then, for
all $\alpha>0$ we have
\begin{equation*}
    \bbe\Big[\log(1+\alpha X_N)\Big]\doteq
    \log\Big(1+\alpha\bbe[X_N]\Big).
\end{equation*}
\end{theorem}

Consider the set of random variables $\{Y_1,\dots,Y_N\}$ where $Y_i\sim
G(y)$ and define $X_i\dff Y^{(N-i+1)}$, where $Y^{(i)}$ is the $i^{th}$
order statistic of the set $\{Y_1,\dots,Y_N\}$; hence, $X_N\sim
G^{(1)}(x)$. Therefore, as provided in \cite[Sec. 4.6]{Arnold:Book}
\begin{eqnarray*}
    \mu_N &\dff& \bbe[X_N]=\sum_{n=1}^N\frac{1}{n}\\
    \sigma^2_N &\dff& \bbe[|X_N-\mu_N|^2]=\sum_{n=1}^N\frac{1}{n^2},
\end{eqnarray*}
which confirms that for finite $N$, $\mu_N$ is also finite. Also as shown in \cite{Sanayei:IT07}, $\mu_N\doteq\log N$ and
$\sigma^2_N\doteq\frac{\pi^2}{6}$, from which it is concluded that
$\frac{\sigma_N}{\mu_N}\rightarrow 0$ as $N\rightarrow\infty$. Therefore
the conditions of the above theorem are satisfied and we have
\begin{align*}
    \int_0^{\infty}\log(1+\alpha x)G^{(1)}dx&=\bbe\Big[\log(1+\alpha
    X_N)\Big]\\
    &\doteq\log\Big(1+\alpha\bbe[X_N]\Big)\\
    &=\log(1+\alpha\log N)\\
    &\doteq\log\log N+\log a,
\end{align*}
which is the desired result.

\section{}\label{app:lower}
By using the following lemma, we further find a lower bound on $R_m$
which will be more mathematically tractable.
\begin{lemma}
\label{lemma:condition}

For a continuous random variable $X$, increasing function $g(\cdot)$ and
real values $b\geq a$
\begin{equation*}
    \bbe\Big[g(X) \med X\geq b \Big]\geq \bbe\Big[g(X) \med X\geq
    a\Big].
\end{equation*}
\end{lemma}
\begin{proof}
See Appendix \ref{app:lemma:condition}.
\end{proof}
By recalling the definition of $\sinr_m^{(i)}$ we have
\begin{align}
    \nonumber &\sum_{i\in\mathcal{H}_m}\bbe\bigg[\log\Big(1+\sinr_{m,i}\Big)\;\bigg|\;
    \mathcal{H}_m\bigg]\\
    \nonumber=& \sum_{i\in\mathcal{H}_m}\bbe\bigg[\log\Big(1+\sinr_{m,i}\Big)\;\bigg|\;
    \sinr_{m,i}\geq\lambda(m,i)\;\\
    \label{eq:Rm_lower0}&\hspace{.9 in};\forall m'\neq m:
    \frac{\sinr_{m,i}}{\lambda(m,i)}\geq \frac{\sinr_{m',i}}{\lambda(m',i)}\bigg]\\
    \label{eq:Rm_lower1} \geq& \sum_{i\in\mathcal{H}_m}\bbe\bigg[\log\Big(1+\sinr_{m,i}\Big)\;\bigg|\;
    \sinr_{m,i}\geq\lambda(m,i)\bigg]\\
    \label{eq:Rm_lower2} \geq &\sum_{i\in\mathcal{H}_m}\bbe\bigg[\log\Big(1+\sinr_{m,i}\Big)\;\bigg|\;
    \sinr_{m,i}\geq\min_{j}\lambda(m,j)\bigg]\\
    \nonumber=& \sum_{i\in\mathcal{H}_m}\bbe\bigg[\log\Big(1+\sinr_{m,i}\Big)\;\bigg|\;
    \sinr_{m,i}\geq\min_{j}\lambda(m,j)\\
    \label{eq:Rm_lower3}&\hspace{.8 in}; \forall l\notin\mathcal{H}_m: \sinr_{l,m}<\min_{j}\lambda(m,j)
    \bigg]\\
    \label{eq:Rm_lower4}=& \sum_{j=1}^{|\mathcal{H}_m|}\bbe\bigg[\log\Big(1+\sinr_m^{(j)}\Big)\;\bigg|\;
    \sinr_m^{(j)}\geq\min_{i}\lambda(m,i)\bigg]\\
    \label{eq:Rm_lower5} \geq &
    \sum_{j=1}^{|\mathcal{H}_m|}\bbe\bigg[\log\Big(1+\sinr_m^{(j)}\Big)\;\bigg],
\end{align}
where (\ref{eq:Rm_lower0}) is obtained by replacing $\mathcal{H}_m$ by
an equivalent representation. Transition from (\ref{eq:Rm_lower0}) to
(\ref{eq:Rm_lower1}) holds by applying Lemma~\ref{lemma:condition} for
$b=\sinr_{m',i}\cdot\frac{\lambda(m,i)}{\lambda(m,i')}$ and $a=0$ for
all $m'\neq m$. Transition to (\ref{eq:Rm_lower2}) is again justified
by using Lemma~\ref{lemma:condition}. Due to the statistical
independence of $\sinr_{m,i}$ and $\sinr_{m,l}$ for $m\in\mathcal{H}_m$
and $l\notin\mathcal{H}_m$ the additional constraints imposed in
(\ref{eq:Rm_lower3}) do not result in any changes. The conditions in
(\ref{eq:Rm_lower3}) are equivalent to having the $|\mathcal{H}_m|$
largest components of $\{\sinr_{m,n}\}_{n=1}^N$ be greater than
$\min_{1\leq i\leq N}\lambda(i)$, which is mathematically stated in
(\ref{eq:Rm_lower4}). Finally, (\ref{eq:Rm_lower5}) holds by applying
Lemma~\ref{lemma:condition} one more time.

\section{Proof of Lemma \ref{lemma:f}} \label{app:lemma:f}
By the expansion of $\Big(x+(1-x)\Big)^N$ we have
\begin{eqnarray*}
    f(x,i)&=& 1-\sum_{j=i+1}^N{N\choose j}x^{N-j}(1-x)^j\\
    &=& 1-\sum_{j=i+1}^N{N\choose N-j}x^{N-j}(1-x)^j\\
    &=& 1-\sum_{k=0}^{N-(i+1)}{N\choose k}(1-x)^{N-k}x^k\\
    &=&1-f(1-x,N-i-1),
\end{eqnarray*}
where it can be concluded that $f'(u,i)\big|_{u=x}=f'(u,N-i-1)\big|_{u=1-x}$. So it is
sufficient to show that $f'(x,i)\geq 0$ for $x\leq \frac{1}{2}$ and for
all $i=1,\dots, N-1$. For this purpose we consider two cases of
$i\leq\lfloor\frac{N}{2}\rfloor$ and
$i>\lfloor\frac{N}{2}\rfloor$. \\
{\underline{\bf Case 1: $i\leq\lfloor\frac{N}{2}\rfloor$}}
\begin{eqnarray}
  \nonumber f'(x,i) &=& \sum_{j=0}^{i}{N\choose j}(N-j)x^{N-j-1}(1-x)^j\\\
  \nonumber &-&{N\choose
  j}jx^{N-j}(1-x)^{j-1}\\
  \nonumber &=& \sum_{j=0}^{i}{N\choose
  j}x^{N-j-1}(1-x)^{j-1}\Big[N(1-x)-j\Big],
\end{eqnarray}
where since $0\leq j\leq i$ it can be shown that for $x\leq \frac{1}{2}$
\begin{align}\label{eq:lemma:f1}
    \nonumber N(1-x)-j&\geq N(1-x)-i\\
    &\geq
    \nonumber N(1-x)-\frac{N}{2}\\
    \nonumber&=\frac{N}{2}(1-2x)\\
    &\geq 0.
\end{align}
{\underline{\bf Case 2: $i>\lfloor\frac{N}{2}\rfloor$}}\\
Define $a_j=1-\frac{1}{2}\delta(\lfloor\frac{N}{2}\rfloor-j)$, where
$\delta(\cdot)$ is the Dirac delta function. Therefore, we get
\begin{align*}
   f(x,i)&=f(x,N-i-1)\\
   &+\sum_{j=N-i}^{\lfloor\frac{N}{2}\rfloor}a_j{N\choose
   j}\Big[x^{N-j}(1-x)^j+x^j(1-x)^{N-j}\Big].
\end{align*}
For $x\leq \frac{1}{2}$ we get
\begin{align}
   \nonumber f'(x,i) &= f'(x,N-i-1)\\
   \nonumber &+\sum_{j=N-i}^{\lfloor\frac{N}{2}\rfloor}a_j{N\choose j}\bigg\{
   x^{N-j-1}(1-x)^{j-1}\Big[N-j-Nx\Big]\\
   \nonumber &+\underset{\geq
   x^{N-j-1}(1-x)^{j-1}}{\underbrace{x^{j-1}(1-x)^{N-j-1}}}\Big[j-Nx\Big]\bigg\}\\
   \nonumber &\geq f'(x,\underset{\leq \lfloor\frac{N}{2}\rfloor}{\underbrace{N-i-1}})\\
   \nonumber & +
   \sum_{j=N-i}^{\lfloor\frac{N}{2}\rfloor}a_j{N\choose
   j} x^{N-j-1}(1-x)^{j-1}\underset{\geq 0}{\underbrace{\Big[N-2Nx\Big]}}\\
   \label{eq:lemma:f2}&\geq 0.
\end{align}
From~(\ref{eq:lemma:f1})~and~(\ref{eq:lemma:f2}) it is concluded
that for $x\leq \frac{1}{2}$, $f(x,i)$ is an increasing function of
$x$, which completes the proof.

\section{Proof of Lemma \ref{lemma:exp_scaling2}}
\label{app:lemma:exp_scaling2}

This proof follows the same spirit as the analysis provided in
\cite{Sanayei:IT07}. However, due to some differences in our setting,
we provided an independent treatment.

For any given number of users $N$, we define a random variable $X_N$,
distributed as $X_N\sim G^N(x)$ and also for $j=1,\dots, N$ we define
\begin{eqnarray*}
  \mu_{(j)} &\dff& \int_0^{\infty}x\;dG^{(j)}(x),\\
  \mbox{and}\;\;\;\sigma^2_{(j)} &\dff&
  \int_0^{\infty}\Big(x-\mu_{(i)}\Big)^2\;dG^{(j)}(x),\\
  \mbox{and}\;\;\;\mu_N &\dff&
  \bbe[X_N]=\int_0^{\infty}x\;dG^N(x)\\
  &=&\sum_{j=1}^NQ_j
  \int_0^{\infty}x\;dG^{(j)}(x)=\sum_{j=1}^NQ_j\mu_{(j)}.
\end{eqnarray*}
As given in~\cite[Sec. 4.6]{Arnold:Book} and discussed in details
in~\cite{Sanayei:IT07}, for ordered exponentially distributed random
variables $X_N$ we have
\begin{equation}
    \label{eq:sigma_n}
    \sigma^2_N<2+2\mu_{(1)}\Big(\mu_{(1)}-\mu_{N}\Big),
\end{equation}
\begin{equation}
    \label{eq:mu_1}
    \mbox{and}\;\;\;\log N+\zeta+\frac{1}{2(N+1)}\leq\mu_{(1)}\leq\log
    N+\zeta+\frac{1}{2N},
\end{equation}
and therefore,
\begin{equation*}
     \mu_{(1)}\doteq \log N,
\end{equation*}
where $\zeta\approx0.577$ is the Euler-Mascheroni constant. Also
\begin{equation*}
     \mu_{(1)}-\log\bigg(\sum_{j=1}^NjQ_j\bigg)-\zeta-0.5\leq\mu_N\leq\mu_{(1)}.
\end{equation*}
By taking into account the constraint in (\ref{eq:Q1}), as
$N\rightarrow\infty$
\begin{equation}
    \label{eq:mu_n}
    1-\frac{\log\bigg(\sum_{j=1}^NjQ_j\bigg)-\zeta-0.5}{\log N}\leq\mu_N\leq 1.
\end{equation}
Equations~(\ref{eq:mu_1})~and~(\ref{eq:mu_n}) together show that
\begin{equation}
    \label{eq:mu_n2}
    \mu_{(1)}\doteq\mu_N\doteq\log N,
\end{equation}
which also implies that $\mu_N\rightarrow\infty$. Taking into
account~(\ref{eq:sigma_n})~ and~(\ref{eq:mu_n2}) we also conclude that
$\lim_{N\rightarrow\infty}\frac{\sigma_N}{\mu_N}=0$ and therefore the
conditions of Theorem~\ref{th:limit} are met. Hence, from
Theorem~\ref{th:limit}
\begin{eqnarray*}
    \int_0^{\infty}\log(1+ax)\;dG^N(x)&=&\bbe\Big[\log(1+a X_N)\Big]\\
    &\doteq&
    \log\Big(1+a\bbe[X_N]\Big)\\
     &=&\log\Big(1+a\mu_N\Big)\\
     &\doteq&\log\log N+\log(a).
\end{eqnarray*}

\section{Proof of Lemma~\ref{lemma:condition}}
\label{app:lemma:condition}
\begin{align*}
   \bbe\Big[g(X) \med X\geq b \Big] &= \int_b^{\infty}g(x)f_{X|X\geq b}(x)\;dx \\
   &= \frac{1}{\pr(X\geq b)}\int_b^{\infty}g(x)f_X(x)\;dx   \\
   &\geq\bigg[1-\frac{\pr(X\geq b)}{\pr(X\geq a)}\bigg]g(b)\\
   &+\frac{1}{\pr(X\geq a)}\int_b^{\infty}g(x)f_X(x)dx \\
   &=\frac{g(b)}{\pr(X\geq a)}\;\pr(a\leq X\leq b) \\
   &+\frac{1}{\pr(X\geq
   a)}\int_b^{\infty}g(x)f_X(x)dx\\
   &\geq \frac{1}{\pr(X\geq a)}\int_a^bg(x)f_X(x)\;dx\\
   &+\frac{1}{\pr(X\geq
   a)}\int_b^{\infty}g(x)f_X(x)dx\\
   &=\int_a^{\infty}g(x)f_{X|X\geq a}(x)\;dx\\
   &=\bbe\Big[g(X) \med X\geq
   a
   \Big].
\end{align*}

\renewcommand\url{\begingroup\urlstyle{rm}\Url}

\bibliographystyle{IEEEtran}
\bibliography{IEEEabrv,CR_MUD}

\end{document}